\documentclass[11pt]{article}
\usepackage[svgnames]{xcolor}

\usepackage{graphicx}
\usepackage[utf8]{inputenc}
\usepackage[T1]{fontenc}
\usepackage{stmaryrd}
\usepackage{textcomp,amssymb,amsmath,amsthm,bbm}
\usepackage[english]{babel}
\usepackage[normalem]{ulem}
\usepackage{float} 

\usepackage{paralist}

\usepackage{xcolor}

\usepackage[a4paper,top=3cm,bottom=2cm,left=3cm,right=3cm,marginparwidth=1.75cm]{geometry}

\usepackage{amsmath}
\usepackage{amssymb}
\usepackage{graphicx}
\usepackage{authblk}
\usepackage[colorlinks=true, allcolors=blue]{hyperref}
\usepackage{cleveref}

\newtheorem{theorem}{Theorem}

\newtheorem{lemma}{Lemma}

\newenvironment{proofof}[1]{\begin{proof}[Proof of #1]}{\end{proof}}

\makeatletter
\newtheorem*{rep@theorem}{\rep@title}
\newcommand{\newreptheorem}[2]{%
\newenvironment{rep#1}[1]{%
 \def\rep@title{#2 \ref{##1}}%
 \begin{rep@theorem}}%
 {\end{rep@theorem}}}
\makeatother
\newreptheorem{theorem}{Theorem}

\makeatletter
\newtheorem*{rep@lemma}{\rep@title}
\newcommand{\newreplemma}[2]{%
\newenvironment{rep#1}[1]{%
 \def\rep@title{#2 \ref{##1}}%
 \begin{rep@lemma}}%
 {\end{rep@lemma}}}
\makeatother
\newreplemma{lemma}{Lemma}

\newcommand{\id}{\sigma_0}
\newcommand{\sigmagreedy}{\sigma_{\sf alg}}

\title{A Short Proof of Convexity of Step-out Step-in Sequencing Games}

\author{Coulter Beeson ~ Neil Olver \\ Department of
Mathematics, London School of
Economics and Political Science, United Kingdom
}

\date{}

\begin{document}
\maketitle

\begin{abstract}
    The Step out--Step in sequencing game is a particular example of a game from the sequencing game framework of Curiel, Perderzoli, and Tijs, where coalitions of players in a queue may reorder themselves to improve the their overall cost, under some restrictions.
    Musegaas, Borm and Quant proved, in two papers, that a simple greedy algorithm correctly computes the valuation of a coalition, and that the game is convex.
    These proofs entail rather involved case analyses;
    in this note, we give short proofs of both results.
\end{abstract}

\section{Introduction}

There are many natural settings where people or tasks form a queue and incur some cost for how long they wait. Cureil, Pederzoli, and Tijs~\cite{curiel1989sequencing} took a cooperative game theory perspective and asked the questions ``how best can a coalition of players re-arrange themselves to save costs?'' and ``how can these cost savings be shared?'' 
The main assumption is that players outside of this coalition should be no worse off. How one interprets the notion of ``no worse off'' gives rise to many different variants of sequencing games. 
Musegaas, Brom, and Quant~\cite{musegaas2015step} introduced the Step out--Step in (SoSi) sequencing game where coalition members cannot move ahead of non-coalition members. They provide a greedy algorithm for computing the optimal cost savings for a coalition in their original paper~\cite{musegaas2015step}.
In a follow-up work, they prove that the SoSi game is convex~\cite{musegaas2018convexity}.
This has important implications for the cooperative game, e.g., that the core is nonempty and can be efficiently computed when the characteristic function can be computed efficiently~\cite{shapley1971cores}.
Both of these proofs are rather lengthy, and involve a large number of case distinctions to complete the argument.

In this paper we provide greatly simplified proofs of the convexity of SoSi games and the correctness of the greedy algorithm. 
Throughout we assume familiarity with the basic concepts of cooperative game theory (see, e.g., \cite{peleg2007introduction}). 
We start by formally introducing SoSi games along with other preliminaries.
The proofs follow in \Cref{sec:proofs}.

\subsection{Single machine scheduling and sequencing games}
Consider the very classical setting of minimizing the weighted sum of completion times, on just a single machine.
We are given a set $N$ of $n$ players, processing times $p \in \mathbb{R}^N_{\ge0}$, and weights $w \in \mathbb{R}^N_{\geq 0}$.
Given an ordering $\sigma: N \to \{1,2,\ldots, n\}$ of the jobs, 
job $j$ will complete at time $c_j(\sigma) := \sum_{k \in N: \sigma(k) \leq \sigma(j)} p_j$, and incur a cost of $C_j(\sigma) := w_j c_j(\sigma)$.
The goal is to minimize the total cost $\sum_{j \in N} C_j(\sigma)$.
The celebrated \emph{Smith's rule}~\cite{Smith1956} states 
that a minimum cost solution is obtained by ordering the jobs by decreasing \emph{urgency} $u_j := w_j/p_j$.

In a \emph{sequencing game}, this machine scheduling setting is augmented by an initial ordering $\sigma_0$ of the players.
Further, a set of \emph{admissible orders} available to a coalition $S$, denoted by $\mathcal{A}(S, \sigma_0)$, is described in some way.
Then the \emph{value} of a coalition $S$ is
\[ v(S) := C_S(\id) - \min_{\sigma^* \in \mathcal{A}(S, \sigma_0)} C_S(\sigma^*), \]
where $C_S(\sigma) := \sum_{i \in S} C_i(\sigma)$ denotes the total cost of players in coalition $S$ under the given ordering.
In the SoSi game, an order $\sigma$ is admissible for a coalition $S$ if for all $i \notin S$ and $j \in N$ with $\id(i) \leq \id(j)$, we have $\sigma(i) \leq \sigma(j)$. 

Musegaas et al.~\cite{musegaas2015step} consider the following greedy algorithm for computing $v(S)$ 
(in fact, they consider a variation of this which is more complicated to describe, but which is equivalent; see \Cref{lem:greedyfacts}). 
Start with the order $\sigma' = \id$, and consider each player $i \in S$ in turn, from latest to earliest according to $\sigma_0$.
When considering player $i$, we update $\sigma'$ by moving the player to a position \emph{later} in the ordering that yields the greatest cost savings for $S$ (breaking ties by choosing the earliest optimal position); if no such move yields an improvement, then $\sigma'$ is left unchanged.
After all players in $S$ have been considered, the algorithm returns $C_S(\id) - C_S(\sigma')$ as the value of coalition $S$.

They prove (in two different works) the following theorems. 
\begin{theorem}\label{greedy}\cite{musegaas2015step}
The greedy algorithm correctly computes $v(S)$.
\end{theorem}
\begin{theorem}\label{convex}\cite{musegaas2018convexity}
 The step-out step-in sequencing game is convex. 
\end{theorem}

Their proof of convexity relies on the correctness of the greedy
algorithm; ours will also. Our proof of Theorem 2, once Theorem 1 has been obtained, is particularly short and simple.

\section{Proofs}\label{sec:proofs}

\paragraph{Preliminaries.}
We start by defining some useful notation. 
For any two players $j,k \in N$ let
\[
    \delta_{jk}(S):=\mathbbm{1}_{S}(k)p_j w_k-\mathbbm{1}_{S}(j)p_kw_j,
\]
where $\mathbbm{1}_{S}$ is the indicator function for the set $S$.
The interpretation of this is the decrease in cost for a coalition $S$ obtained by swapping $j$ and $k$, given that in the current order $j$ directly precedes $k$. Note that if $u_i > u_j$, then $\delta_{ij}<0$, and $\delta_{ij}=-\delta_{ji}$.
We also define, for $P \subset N$ and $j \in N \setminus P$, 
\[ \delta_{jP}(S):=\sum_{k\in P} \delta_{jk}(S); \]
if $P$ is a contiguous sets of elements in some order, with $j$ adjacent to and preceding $P$, this is the cost decrease obtained by moving $j$ to immediately after $P$. Similarly $\delta_{Pj}(S) := -\delta_{jP}(S)$ is the cost decrease obtained by moving $j$ from immediately after $P$ to immediately before $P$.
We will sometimes omit the explicit dependence on $S$ (writing, for example, $\delta_{ij}$ rather than $\delta_{ij}(S)$) when it can cause no ambiguity.

The subgame restricted to the players $T\subset N$, whose value function we denote by $v^T$, is the sequencing game arising from just the players $T$ with the same weights and processing times, and initially ordered according to the same relative order as $\sigma_0$.
For notational convenience, we define $v^T(S) := v^T(S \cap T)$ for all $S \subseteq N$.

For a given coalition $S$ and an order $\sigma$, a \emph{component} refers to a maximal subset of $S$ that is contiguous with respect to $\sigma$; that is, an inclusion-wise maximal set of the form $\{ j \in N: \sigma(i) \leq \sigma(j) \leq \sigma(k)\} \subseteq S$ for some $i,k\in S$.

\paragraph{The greedy algorithm.}
\begin{lemma}\label{lem:swap}
    For any distinct $j,k \in S$ with $u_j \leq u_k$ and $\ell \in N$, $w_j^{-1}\delta_{j\ell} \geq w_k^{-1}\delta_{k\ell}$.
\end{lemma}
\begin{proof}
    If $\ell \notin S$, then $w_{j}^{-1}\delta_{j\ell} = -p_{\ell} = w_{k}^{-1} \delta_{k\ell}$.
    If $\ell \in S$, then 
    \[ w_j^{-1}\delta_{j\ell} = w_{\ell} \cdot\tfrac{p_j}{w_j} - p_{\ell} \geq w_{\ell} \cdot \tfrac{p_k}{w_k} - p_{\ell} = w_k^{-1} \delta_{k\ell}. \qedhere \]
\end{proof}

Let us say that two orders $\sigma$ and $\sigma'$ are \emph{equivalent} for a coalition $S$ if $\sigma'$ can be obtained from $\sigma$ via a sequence of swaps of adjacent players in $S$ with equal urgency.
Since $\delta_{jk}(S) = p_jp_k(u_k - u_j) = 0$ for $j,k \in S$ with $u_j = u_k$, it follows that $C_S(\sigma) = C_S(\sigma')$ for equivalent orderings $\sigma$ and $\sigma'$.
\begin{lemma}\label{lem:greedyfacts}
    Let $\sigma'_0$ be any order obtained from $\sigma_0$ by rearranging only players within the same component of $S$.
    Let $\sigmagreedy$ and $\sigmagreedy'$ be the orders determined by the greedy algorithm starting from $\sigma_0$ and $\sigma'_0$ respectively.
    Then
    \begin{enumerate}[(i)] 
        \item $\sigmagreedy$ and $\sigmagreedy'$ are equivalent for coalition $S$.
        \item Any two players in the same component of $S$ respect Smith's rule in both $\sigmagreedy$ and $\sigmagreedy'$
            (i.e., are correctly ordered by urgency).
    \end{enumerate}
\end{lemma}
\begin{proof}
    It suffices to show that the claim holds if $\sigma'_0$ differs from $\sigma_0$ by swapping two adjacent players $j,k \in S$, with $\sigma_0(j) < \sigma_0(k)$, since any valid reordering can be obtained via a sequence of such swaps.
    By symmetry of $\sigma_0$ and $\sigma'_0$ we may assume $u_j \geq u_k$.
    Let $W := \{ i \in N : \sigma_0(k) < \sigma_0(i) \}=\{ i \in N : \sigma'_0(j) < \sigma'_0(i) \}$, let $\sigma^W$ be order returned by greedy algorithm for the sub-problem on $W$. 
    Let $\emptyset = W_0 \subsetneq W_1 \subsetneq \cdots \subsetneq W_r = W$ be all prefixes of $W$ under $\sigma^W$ 
    (i.e., all distinct sets of the form $\{ x \in N: \sigma^W(x) \leq \sigma^W(y)\}$ for some $y$, along with the empty set).

    Let $t^*$ be such that $\delta_{k W_{t^*}}$ is maximized, and $s^*$ be such that $\delta_{j W_{s^*}}$ is maximized (breaking ties by choosing $t^*$ and $s^*$ as small as possible).
    Consider any $t > t^*$, and let $Q = W_t \setminus W_{t^*}$. 
    By the choice of $t^*$, $\delta_{k Q} = \delta_{k W_t} - \delta_{k W_{t^*}} \leq 0$.
    Thus
    \begin{align*}
        \delta_{j W_t} &= \delta_{j W_{t^*}} + \delta_{j Q}\\
                       &\leq \delta_{j W_{t^*}} + \frac{w_j}{w_k} \delta_{k Q} && \text{ (by \Cref{lem:swap})}\\
                       &\leq \delta_{j W_{t^*}}.
    \end{align*}
So $s^* \leq t^*$. 

    Since also $\delta_{jk} \leq 0$,
    it follows that when running the greedy algorithm from $\sigma_0$, $j$ will be placed before $k$ in the final ordering.
    On the other hand, starting from $\sigma_0'$, $j$ will first be placed immediately to the right of $W_{s^*}$, and $k$ will be placed to the right of $j$ if $u_j > u_k$, or immediately to the left of $j$ if $u_j = u_k$ (since then $s^* = t^*$ and $\delta_{jk} = 0$).
    That is, the orders obtained by the greedy algorithm after $j$ and $k$ have been considered starting from $\sigma_0$ and $\sigma'_0$ are equivalent; call these orders $\sigma$ and $\sigma'$ respectively.
    Note that this also proves the second part of the lemma for $j$ and $k$ adjacent.
    
    It suffices to show that inserting the remaining players into $\sigma$ and $\sigma'$ does not violate equivalence. 
    The optimal position to place any player $\ell$ into any two equivalent orderings is the same. 
    If there is a group $P$ of contiguous players with equal urgency, 
    then $\ell$ will either placed ahead of $P$ (if $\delta_{\ell P} \leq 0$ or behind $P$ (if $\delta_{\ell P}>0$), irrespective of the ordering of the players in $P$.

    \medskip

    We now prove the second part of the lemma in generality, where $j$ and $k$ are in the same component but not necessarily adjacent.
    Assume, relabelling if necessary, that $u_j > u_k$ (if $u_j = u_k$, there is nothing to prove).
    Then as already observed, $\sigmagreedy'(j) < \sigmagreedy'(k)$.
    Since $\sigmagreedy$ is equivalent to $\sigmagreedy'$ by the first part of the lemma, $\sigmagreedy(j) < \sigmagreedy(k)$ as well.
\end{proof}

We are now ready to prove that the greedy algorithm correctly computes $v(S)$.
\begin{proofof}{\Cref{greedy}}
    By \Cref{lem:greedyfacts}, it suffices to prove the claim under the assumption that each component of $S$ is ordered in $\sigma_0$ according to Smith's rule.

    Assume inductively that the claim holds for all restrictions of the instance to a strict subset of players.
    If $S=N$, then $\sigma_0$ is already optimal and the greedy algorithm will not change it, so the claim is clear.
    So assume $S \neq N$, and let $q$ be the earliest player under $\sigma_0$ not in $S$. 
Let $P \subset S$ be the players preceding $q$, and $W:= N \setminus \bigl( P \cup \{q\} \bigr)$ be all players after $q$.

Denote by $\sigma$ the ordering at the point just before any player in $P$ is considered, and by $\sigmagreedy$ the final ordering determined by the greedy algorithm.
Also let $\emptyset = W_0 \subsetneq W_1 \subsetneq \cdots \subsetneq W_r = W$ be all the prefixes of $W$ under $\sigma$.
By \Cref{lem:greedyfacts}, the order of the players in $P$ is maintained in $\sigmagreedy$. %
This means that the cost decrease (if any) attributable to moving player $i \in P$ is $\max\{0, \gamma_i\}$, where
\[ \gamma_i := \delta_{iq} + \max_{0 \leq k \leq r} \delta_{iW_k}. \]
In other words, either we don't move $i$ at all and there is no change, or we obtain a cost decrease by moving $i$ to some location after $q$.
So
\begin{equation}\label{eq:greedycost}
    C_S(\sigma_0) - C_S(\sigmagreedy) = C_S(\sigma_0) - C_S(\sigma) + \sum_{i \in P} \max\{0, \gamma_i\}.
\end{equation}

Consider now, for any fixed $T \subseteq P$, the instance restricted to $W \cup T$.
The greedy algorithm will again first order $W$ in exactly the same way as $\sigma$ 
and then insert each player in $T$ into their optimal positions.
The cost decrease associated with player $i \in T$ is then
\[ \max_{0 \leq k \leq r} \delta_{i W_k} = \gamma_i - \delta_{iq}.\]
By induction, the result of the greedy algorithm on this instance is optimal, and so we know that 
\begin{equation}\label{eq:greedyopt}
    v^{W \cup T}(S) = C_S(\sigma_0) - C_S(\sigma) + \sum_{i \in T} (\gamma_i - \delta_{iq}). 
\end{equation}

Consider now an optimal ordering $\sigma^*$. 
Let $T^*$ be the subset of jobs in $P$ that are later than $q$ under $\sigma^*$.
Let $\kappa$ be the cost decrease associated with moving players in $T^*$ past players in $P$: so $\kappa = \sum_{j \in T^*, k \in P \setminus T^*: \sigma_0(j) < \sigma_0(k)} \delta_{jk}$.
By assumption $u_j \geq u_k$ and hence $\delta_{jk} \leq 0$ for all the terms in this sum, so $\kappa \leq 0$.
Then we have
\begin{align*}
    v(S) &= v^{W \cup T^*}(S) + \delta_{T^*q} + \kappa \\
         &\leq C_S(\sigma_0) - C_S(\sigma) + \sum_{i \in T^*} \gamma_i &&\text{by \eqref{eq:greedyopt}}\\
          &\leq C_S(\sigma_0) - C_S(\sigmagreedy) &&\text{by \eqref{eq:greedycost}}. 
\end{align*}
Thus the cost decreased obtained by the greedy algorithm is at least $v(S)$, and hence equal to $v(S)$.
\end{proofof}

\paragraph{Convexity.}
Recall that a game is called \emph{convex} if its characteristic function $v$ is \emph{supermodular}, i.e., 
\[ v(S \cup T) + v(S \cap T) \ge v(S) + v(T) \qquad \forall S,T \subseteq N. \]
A function $v$ is \emph{modular} when we have equality in the above.
We will make use of a lemma due to Lov\'asz about supermodular functions. Recall that a set function $f$ is \emph{monotone} if $f(S) \leq f(T)$ for all $S \subseteq T$.
\begin{lemma}\cite{lovasz1983submodular}
\label{lem:lovasz}
Let $f$ and $g$ be supermodular functions, with $f-g$ monotone.
Then $\max \{ f,g\}$ is supermodular. 
\end{lemma}

We remark that in the case $f(\emptyset)=g(\emptyset)=0$ this lemma is trivial as it can be easily argued that $f(S)\ge g(S)$ for all coalitions $S$. However, we will apply this to more general supermodular functions. 

The following lemma is a simple corollary.
\begin{lemma}
\label{lem:prefix}
Let $f_k(S)= \max_{  1 \leq r \leq k} \bigg\{ \sum_{j=1}^r g_{j}(S) \bigg\}$, where each $g_j(S)$ is monotone and supermodular.
Then $f$ is supermodular.
\end{lemma}
\begin{proof}
Let $h_r(S):=\sum_{j=1}^r g_{j}(S)$ for each $r$.
    Being the sum of supermodular functions, $h_r$ is again supermodular.
Then $f_k(S) = \max \big \{ f_{k-1}(S),h_{k}(S) \big \}$.
By Lemma~\ref{lem:lovasz} it suffices to show that $h_k-f_{k-1}$ is monotone. Let $S\subseteq T$; then
\[    
h_k(S) - f_{k-1}(S) = \min_{1\le r\le k} \sum_{j=r}^{k} g_{j}(S) \le \min_{1\le r\le k} \sum_{j=r}^{k} g_{j}(T)  = h_k(T) - f_{k-1}(T),
\]
where the inequality follows by monotonicity of the $g_j$'s.
\end{proof}

We are now ready to prove the convexity of the SoSi game.
\begin{proofof}{\Cref{convex}}
Let $\sigma_{i}$ be the ordering of players when player $i \in S$ is being considered during the greedy algorithm, and let $f_i(S)$ be the cost saving obtained during this step of the greedy algorithm by moving player $i$ to its optimal location later in the ordering. 
By the correctness of the greedy algorithm, we have that $v(S)=\sum_{i \in S} f_i(S)$. 
It suffices then to show that $f_i(S)$ is supermodular for any fixed $i \in S$. 

Let us write $x \preceq y$ if $\sigma_i(x) \le \sigma_i(y)$. Then 
$f_i(S)=\max_{ r \succeq i } \big\{ \sum_{i \prec j \preceq r} \delta_{ij}(S) \big\}$.
For any $j \in N$, $\delta_{ij}(S) = \mathbbm{1}_S(j)p_iw_j - p_jw_i$ is monotone and modular, and hence supermodular.
Thus we can apply Lemma~\ref{lem:prefix}, and it follows that $v(S)$ is supermodular. 
\end{proofof}

\paragraph{Acknowledgements.}
This work was supported by Dutch Science Foundation (NWO) Vidi grant 016.Vidi.189.087.

\bibliographystyle{acm}
\bibliography{sample}

\end{document}